\documentclass[runningheads,a4paper]{llncs}

\usepackage{times}
\usepackage{graphicx}
\usepackage{algorithm, algpseudocode}

\usepackage{url}
\newcommand{\keywords}[1]{\par\addvspace\baselineskip
\noindent\keywordname\enspace\ignorespaces#1}

\newtheorem{fact}[proposition]{Proposition}

\newcommand{\eqref}[1]{(\ref{#1})}

\usepackage{xspace}

\newcommand{\E}{\textbf{E}}

\newcommand{\broadcast}{{\it BROADCAST}}
\newcommand{\broadcastws}{{\it BROADCAST}\xspace}
\newcommand{\send}{{\it SEND}\xspace}
\newcommand{\sendpath}{{\it SEND-PATH}\xspace}
\newcommand{\che}{{\it CHECK}\xspace}
\newcommand{\simplecheck}{{\it CHECK1}\xspace}
\newcommand{\advancedcheck}{{\it CHECK2}\xspace}
\newcommand{\update}{{\it UPDATE}\xspace}

\newcommand{\pcheckcall}{p_{call}}
\newcommand{\pcheckcatch}{p_{detect}}

\newcommand{\chkprobinv}{(\log^{*} n)^{2}}
\newcommand{\numbadnodes}{n/8}

\newcommand{\chkrows}{4\log^{*} n}

\newcommand{\sender}{{\bf s}\xspace}
\newcommand{\receiver}{{\bf r}\xspace}

\newcommand{\llog}{\emph{self-healing}\xspace}
\newcommand{\bfly}{\emph{no-self-healing}\xspace}

\newcommand{\paperTitle}{Self-Healing of Byzantine Faults}

\newcommand{\paperTitleAbbr}{Self-Healing of Byzantine Faults}

\newcommand{\authorsAbbr}{J. Knockel, G. Saad and J. Saia}

\date{}

\begin{document}
\mainmatter  

\author{Jeffrey Knockel \and George Saad \and Jared Saia}

\institute{Department of Computer Science, University of New Mexico,\\
\email{\{jeffk,george.saad,saia\}@cs.unm.edu}
}


\title{\paperTitle\footnote{Eligible for Best Student Paper Award (George Saad is a full-time student).}}

\titlerunning{\paperTitleAbbr}

\authorrunning{\authorsAbbr}
\tocauthor{Jeffrey Knockel, George Saad, and  Jared Saia}

\toctitle{\paperTitleAbbr}
\maketitle

\begin{abstract}

Recent years have seen significant interest in designing networks that are \emph{self-healing} in the sense that they can automatically recover from adversarial attacks. Previous work shows that it is possible for a network to automatically recover, even when an adversary repeatedly deletes nodes in the network.  However, there have not yet been any algorithms that self-heal in the case where an adversary takes over nodes in the network. In this paper, we address this gap.

In particular, we describe a communication network over $n$ nodes that ensures the following properties, even when an adversary controls up to $t \leq (1/8-\epsilon)n $ nodes, for any non-negative $\epsilon$. First, the network provides a point-to-point communication with bandwidth and latency costs that are asymptotically optimal. Second, the expected total number of message corruptions is $O(t \chkprobinv )$ before the adversarially controlled nodes are effectively quarantined so that they cause no more corruptions.    
Empirical results show that our algorithm can reduce the bandwidth cost by up to a factor of $70$.
\keywords{Byzantine Faults, Threshold Cryptography, Self-Healing Algorithms.}
\end{abstract}

{\em ``Fool me once, shame on you. Fool me twice, shame on me.'' - English proverb}

\section{Introduction}

Self-healing algorithms protect critical properties of a network, even when that network is under repeated attack.  Such algorithms only expend resources when it is necessary to repair damage done by an attacker.  Thus, they provide significant resource savings when compared to traditional robust algorithms, which expend significant resources even when the network is not under attack.


The last several years have seen exciting results in the design of self-healing algorithms~\cite{boman2006brief,saia2008picking,hayes2008forgiving,hayes2009forgiving,pandurangan2011xheal,sarma2011edge}.  Unfortunately, none of these previous results handle \emph{Byzantine faults}, where an adversary takes over nodes in the network and can cause them to deviate arbitrarily from the protocol.  This is a significant gap, since traditional Byzantine-resilient algorithms are notoriously inefficient, and the self-healing approach could significantly improve efficiency.

In this paper, we take a step towards addressing this gap.  For a network of $n$ nodes, we design self-healing algorithms for communication that tolerate up to $1/8$ fraction of Byzantine faults.  Our algorithms enable any node to send a message to any other node in the network with bandwidth and latency costs that are asymptotically optimal.  


Moreover, our algorithms limit the expected total number of message corruptions.  Ideally, each Byzantine node would cause $O(1)$ corruptions; our result is that each Byzantine node  causes an expected $O((\log^{*}n)^{2})$ corruptions.
\footnote{Recall that $\log^{*} n$ or the iterated logarithm function is the number of times logarithm must be applied iteratively before the result is less than or equal to $1$.  It is an extremely slowly growing function: e.g. $\log^{*} 10^{10} = 5$.}
Now we must amend our initial proverb to: \emph{``Fool me once, shame on you.  Fool me $\omega((\log^{*}n)^{2})$ times, shame on me.''}.
\subsection{Our Model}

We assume an adversary that is \emph{static} in the sense that it takes over nodes before the algorithm begins. The nodes that are compromised by the adversary are \emph{bad}, and the other nodes are \emph{good}. The bad nodes may arbitrarily deviate from the protocol, by sending no messages, excessive numbers of messages, incorrect messages,  or any combination of these.  The good nodes follow the protocol.  We assume that the adversary knows our protocol, but is unaware of the random bits of the good nodes.

We further assume that each node has a unique ID.  We say that node $p$ has a link to node $q$ if $p$ knows $q$'s ID and can thus directly communicate with node $q$.  Also, we assume the existence of a public key digital signature scheme, and thus a computationally bounded adversary.
Finally, we assume a partially synchronous communication model: any message sent from one good node to another good node requires at most $h$ time steps to be sent and received, and the value $h$ is known to all nodes.  Also, we allow for the adversary to be \emph{rushing}, where the bad nodes receive all messages from good nodes in a round before sending out their own messages.

Our algorithms make critical use of quorums and a quorum graph.  
We define a \emph{quorum} to be a set of $\theta(\log n)$ nodes, of which at most $1/8$-fraction are bad.  
Many results show how to create and maintain a network of quorums~\cite{FS2,hildrum2003,NW03,scheideler:how,FSY,AS4,king2011load}.  All of these results maintain what we will call a \emph{quorum graph} in which each vertex represents a quorum.    
The properties of the quorum graph are: 
1) each node is in $O(\log n)$ quorums; 
2) for any quorum $Q$, any node in $Q$ can communicate directly to any other node in $Q$; and 
3) for any quorums $Q_{i}$ and $Q_{j}$ that are connected in the quorum graph, any node in $Q_{i}$ can communicate directly with any node in $Q_{j}$ and vice versa.
Moreover, we assume that for any two nodes $x$ and $y$ in a quorum, node $x$ knows all quorums that node $y$ is in.

The communication in the quorum graph typically occurs as follows.  
When a node $\sender$ sends another node $\receiver$ some message $m$, there is a canonical \emph{quorum path}, $Q_{1}, Q_{2}, \ldots, Q_{\ell}$, through the quorum graph.
This path is determined by the ID's of both $\sender$ and $\receiver$.  A naive way to route the message is for $\sender$ to send $m$ to all nodes in $Q_{1}$.  
Then for $i = 1$ to $\ell -1$, for all nodes in $Q_{i}$ to send $m$ to all nodes in $Q_{i+1}$, and for each node in $Q_{i+1}$ to do majority filtering on the messages received in order to determine the true value of $m$. 
Then all nodes in $Q_\ell$ send $m$ to node $\receiver$ that does a majority filter on the received messages.
Unfortunately, this algorithm requires $O(\ell \log^{2} n)$ messages and a latency of $O(\ell)$.
This paper shows how to reduce the message cost to $O(\ell + \log n)$, in an amortized sense. 

As we show in Section \ref{s:empirical}, this reduction can be large in practice.  
In particular, we reduce the bandwidth cost by a factor of $58$ for $n = 14{,}116$, and by a factor of $70$ for $n = 30{,}509$.

\subsection{Our Results}

This paper provides a self-healing algorithm, \send, that sends a message from a source node to a target node in the network.  Our main result is summarized in the following theorem.

\begin{theorem}\label{thm:corruptions}
Assume we have a network with $n$ nodes and $t\le (1/8-\epsilon)n$ bad nodes, for any non-negative $\epsilon$, and a quorum graph as described above.  Then our algorithm ensures the following.
\begin{itemize}
\item For any call to \send, the expected number of messages is $O(\ell + \log n)$ and the expected latency is $O(\ell)$, in an amortized sense.\footnote{In particular, if we perform any number of message sends through quorum paths, where $\ell_{M}$ is the longest such path, and $\mathcal{L}$ is the sum of the quorums traversed in all such paths, then the expected total number of messages sent will be $O(\mathcal{L} + t \cdot (\ell_{M} \log^{2} n + \log^5 n))$.  Note that, since $t$ is fixed, for large $\mathcal{L}$, the expected total number of messages is $O(\mathcal{L})$.}
\item The expected total number of times that \send fails to deliver a message reliably is at most $3t \chkprobinv$.
\end{itemize}
\end{theorem}

The proof of this theorem is deferred to appendix \ref{app:proofoftheorem} due to space constraints.

\subsection{Related Work}
Several papers \cite{Frisanco:1997,Iraschko:1998,Murakami:1997,Caenegem:1997,Xiong:1999} have discussed different restoration mechanisms to preserve network performance by adding capacity and rerouting traffic streams in the presence of node or link failures. 
They present mathematical models to determine global optimal restoration paths, and provide methods for capacity optimization of path-restorable networks.


Our results are inspired by recent work on self-healing algorithms~\cite{boman2006brief,saia2008picking,hayes2008forgiving,hayes2009forgiving,pandurangan2011xheal,sarma2011edge}.
A common model for these results is that the following process repeats indefinitely: an adversary deletes some nodes in the network, and the algorithm adds edges.  The algorithm is constrained to never increase the degree of any node by more than a logarithmic factor from its original degree.  In this model, researchers have presented algorithms that ensure the following properties: the network stays connected and the diameter does not increase by much~\cite{boman2006brief,saia2008picking,hayes2008forgiving}; the shortest path between any pair of nodes does not increase by much~\cite{hayes2009forgiving}; and expansion properties of the network are approximately preserved~\cite{pandurangan2011xheal}.

Our results are also similar in spirit to those of Saia and Young~\cite{saia2008reducing} and Young et al.~\cite{young}, which both show how to reduce message complexity when transmitting a message across a quorum path of length $\ell$.  
The first result,~\cite{saia2008reducing}, achieves expected message complexity of $O(\ell \log n)$ by use of bipartite expanders.  However, this result is impractical due to high hidden constants and high setup costs.  The second result,~\cite{young}, achieves expected message complexity of $O(\ell)$.  
However, this second result requires the sender to iteratively contact a member of each quorum in the quorum path.

As mentioned earlier, several peer-to-peer networks have been described that provably enable reliable communication, even in the face of adversarial attack~\cite{FS2,datar2002butterflies,hildrum2003,NW03,scheideler:how,AS4}.  To the best of our knowledge, our approach applies to each of these  networks, with the exception of~\cite{datar2002butterflies}.  In particular, we can apply our algorithms to asymptotically improve the efficiency of the peer-to-peer networks from~\cite{FS2,hildrum2003,NW03,scheideler:how,AS4}. 

Similarly to Young et al. \cite{Max:2012}, we use threshold cryptography as an alternative to Byzantine Agreement.

\subsection{Organization of Paper}

The rest of this paper is organized as follows.  In Section~\ref{s:alg}, we describe our algorithms.
The analysis of our algorithms is shown in Section \ref{s:analysis}.
Section~\ref{s:empirical} gives empirical results showing how our algorithms can improve the efficiency of the butterfly networks of~\cite{FS2}.  Finally, we conclude and describe problems for future work in Section~\ref{s:conc}.

\section{Our Algorithms} \label{s:alg}

In this section, we describe our algorithms: \broadcast, \send, \sendpath, \che and \update.  

\subsection{Overview}

Recall that when node $\sender$ wants to send a message to a node $\receiver$, there is a canonical \emph{quorum path} $Q_1, Q_2, \dots, Q_\ell$, determined by the IDs of $\sender$ and $\receiver$. 
We assume that $Q_1$ is the leftmost quorum and $Q_\ell$ is the rightmost quorum; and we let $|Q_j|$ be the number of nodes in quorum $Q_j$, for $1 \leq j \leq \ell$.

The objective of our algorithms is marking all bad nodes in the network, where no more message corruption occurred.
In order to do that, we mark nodes after they are \emph{in conflict}, where a pair of nodes is \emph{in conflict} if at least one of these nodes accuses the other node in this pair.
If the half of nodes in any quorum are marked, we unmark these nodes.
Note that when we mark (or unmark) a node, it is marked (or unmarked) in its quorums and in their neighboring quorums.
Note that all nodes in the network are initially unmarked.

\begin{algorithm}[h]
\caption{$\broadcast (m, S)$}\label{a:broadcast}
\label{a:sign}
\textbf{Declarations:} \broadcastws is being called by a node $x$ in a quorum $Q$ in order to send a message $m$ to a set of nodes $S$.

\begin{enumerate}
\item Node $x$ sends message $m$ to all nodes in $Q$.
\item Each node in $Q$ signs $m$ by its private key share to obtain a signed-message share, and sends this signed-message share back to node $x$.
\item Node $x$ interpolates at least $\frac{7|Q|}{8}$ of the received signed-message shares to obtain the signed-message of $Q$.
\item Node $x$ sends the signed-message of $Q$ to all nodes in $S$.
\end{enumerate}
\end{algorithm}

\begin{algorithm}[h]
\caption{$\textsc{SEND}(m,\receiver)$}
\label{a:send}
\textbf{Declaration:}   
node $\sender$ wants to send message $m$ to node $\receiver$.
\begin{enumerate}
\item Node $\sender$ calls \sendpath($m,\receiver$).
\item With probability $\pcheckcall$, node $\sender$ calls \che ($m, \receiver$).
\end{enumerate}
\end{algorithm}

\begin{algorithm}[h]
\caption{$\textsc{SEND-PATH}(m,\receiver)$}
\label{a:randnodepath}
\textbf{Declarations:} 
$m$ is the message to be sent, and $\receiver$ is the destination. 


\begin{enumerate}
\item Node $\sender$ selects an unmarked node, $q_1$, in $Q_1$ uniformly at random. 
\item Node $\sender$ broadcasts to all nodes in $Q_1$ the message $m$ and $q_1$'s ID.
\item All nodes in $Q_1$ forward the message to node $q_1$.
\item For $i = 1, \ldots, \ell-1$ do
\begin{enumerate}
\item Node $q_i$ selects an unmarked node, $q_{i+1}$, in $Q_{i+1}$ uniformly at random. 
\item Node $q_i$ sends $m$ to node $q_{i+1}$.
\end{enumerate}
\item Node $q_{\ell}$ broadcasts $m$ to all nodes in $Q_\ell$.
\item All nodes in $Q_\ell$ send $m$ to node $\receiver$.
\end{enumerate}
\end{algorithm}

In our algorithms, we assume that when any node $x$ broadcasts a message $m$ to a set of nodes $S$, it executes $\broadcast (m, S)$. 
Before discussing our main algorithm, \send, we describe \broadcastws procedure (Algorithm \ref{a:broadcast}).

In \broadcast, we make use of the threshold cryptography as an alternative to Byzantine Agreement. 
We briefly describe the threshold cryptography, ($\eta,d$)-threshold scheme.

\medskip
\noindent
\textbf{Threshold Cryptography.} In ($\eta,d$)-threshold scheme, the secret key is distributed among $\eta$ parties, and any subset of more than $d$ parties can jointly reassemble the key. 
The secret key can be distributed by a completely distributed approach, \emph{Distributed Key Generation} (DKG) \cite{Kate:2009}.
The \emph{Distributed Key Generation} (DKG) generates the public/private key shares of all nodes in every quorum and the public/private key pair of each quorum.
The public key share of any node is known only to the nodes that are in the same quorum.
Moreover, the public key of each quorum is known to all nodes of this quorum and its neighboring quorums, but the private key of any quorum is unknown to all nodes.

\medskip
\noindent
\textbf{BROADCAST Algorithm.} In $\broadcast (m, S)$, we use in particular a ($|Q|, \frac{7}{8}|Q|-1$)-threshold scheme, where $|Q|$ is the quorum size.
When a node $x$ in a quorum $Q$ calls $\broadcast (m, S)$ to broadcast a message $m$ to a set of nodes $S$, firstly node $x$ sends $m$ to all nodes in $Q$. Then each node in the quorum signs $m$ by its private key share, and sends the signed message share to node $x$. 
Node $x$ interpolates at least $\frac{7|Q|}{8}$ of the received signed-message shares to construct a signed-message of the quorum.
We know that at least $7/8$-fraction of the nodes in any quorum are good. 
So if this signed-message is constructed, it ensures that at least $7/8$-fraction of the nodes in this quorum has received the same message $m$, agrees upon the content of the message and gives the permission to node $x$ of broadcasting this message.
Then node $x$ sends this constructed signed-message to all nodes in $S$.

%
%
%

\medskip
Now we describe our main algorithm, \send, that is stated formally in Algorithm~\ref{a:send}.
\send calls \sendpath, which is formally described in Algorithm~\ref{a:randnodepath}.
In \sendpath, node $\sender$ sends message $m$ to node $\receiver$ through a path of unmarked nodes selected uniformly at random.
However, \sendpath algorithm is vulnerable to corruption.  
Thus, with probability $\pcheckcall$, \send calls \che, which detects if a message was corrupted in the previous call to \sendpath, with probability $\pcheckcatch$.

In \che, the message is propagated from the leftmost quorum to the rightmost quorum in the quorum path through a path of subquorums, where a subquorum is a subset of unmarked nodes selected uniformly at random in a quorum.

\che is implemented as either \simplecheck (Algorithm~\ref{a:simplecheck}) or \advancedcheck (Algorithm~\ref{a:check}). 
For \simplecheck, $\pcheckcall = 1/(\log\log n)^2$ and $\pcheckcatch = 1-o(1)$, and it requires $O(\ell (\log \log n)^2 + \log n \cdot \log\log n)$ messages and has latency $O(\ell)$. 
For \advancedcheck, $\pcheckcall = 1/\chkprobinv$ and $\pcheckcatch \geq 1/2$, and it has message cost $O((\ell + \log n) (\log^* n)^2)$ and latency $O(\ell \log^* n)$.

Unfortunately, while \che can determine if a corruption occurred, it does not specify the location where the corruption occurred.  
Thus, if \che detects a corruption, \update (Algorithm~\ref{a:update}) is called.   
When \update is called, it identifies two neighboring quorums $Q_{i}$ and $Q_{i+1}$ in the path, for some $1 \leq i < \ell$, such that at least one pair of nodes in these quorums is in conflict and at least one node in such pair is bad. 
Then quorums $Q_{i}$ and $Q_{i+1}$ mark these nodes and notify all other quorums that these nodes are in.
All quorums in which these nodes are notify their neighboring quorums.
Recall that if the half of nodes in any quorum have been marked, these nodes are set unmarked in all quorums they are in, and their neighboring quorums are notified.

Moreover, we use \broadcastws in \sendpath and \che to handle any accusation in \update against node $\sender$ or node $\receiver$.
In \sendpath (or \che), we make node $\sender$ broadcast the message to all nodes in $Q_1$ which forward the message to the selected unmarked node (or subquorum) in $Q_1$; and when the message is propagated to the selected unmarked node (or subquorum) in $Q_\ell$, the message is broadcasted to all nodes in $Q_\ell$.

Our model does not directly consider concurrency.  
In a real system, concurrent calls to UPDATE that overlap at a single quorum may allow the adversary to achieve multiple corruptions at the cost of a single marked bad node.  However, this does not effect correctness, and, in practice, this issue can be avoided by serializing concurrent calls to
SEND.  For simplicity of presentation, we leave the concurrency aspect out of this paper.

Throughout this paper, we let $U_i$ be the set of unmarked nodes in $Q_i$, and we let $|U_i|$ be the number of nodes in $U_i$, for $1\leq i \leq \ell$.



\subsection{\simplecheck}\label{sec:simple-check}

\begin{algorithm}[h]
\caption{$\textsc{CHECK1} (m, \receiver)$}
\label{a:simplecheck}


\textbf{Initialization:} let the subquorums $S_1, S_2,...,S_\ell$ be initially empty.

\begin{enumerate}

\item Node $\sender$ constructs $R$ to be an $\ell$ by $|Q_{max}|$ by $\log \log n$ array of random numbers, where $|Q_{max}|$ is the maximum size of any quorum and $\log \log n$ is the size of any subquorum.
Note that $R[j,k]$ is a multiset of $\log\log n$ numbers selected uniformly at random with replacement between $1$ and $k$.

\item Node $\sender$ sets $m'$ to be a message consisting of $m$, $\receiver$, and $R$.

\item Node $\sender$ broadcasts $m'$ to all nodes in $Q_1$.

\item The nodes in $Q_1$ calculate the nodes of $S_1$ using the numbers in $R[1,|U_1|]$ to index $U_1$'s nodes sorted by their IDs.

\item The nodes in $Q_1$ send $m'$ to the nodes of $S_1$.

\item For $j \gets 1, \ldots, \ell-1$ do
\begin{enumerate}
\item The nodes of $S_j$ calculate the nodes of $S_{j+1}$ using the numbers in $R[j+1, |U_{j+1}|]$ to index $U_{j+1}$'s nodes sorted by their IDs.
\item The nodes of $S_j$ send $m'$ to all nodes of $S_{j+1}$.
\end{enumerate}
\item The nodes of $S_{\ell}$ broadcast $m'$ to all nodes in $Q_\ell$.
\end{enumerate}
\textbf{Note that:} 
throughout \simplecheck, if a node receives inconsistent messages or fails to receive an expected message, then it initiates a call to \update.  
\end{algorithm}

Now we describe \simplecheck that is stated formally as Algorithm \ref{a:simplecheck}.
\simplecheck is a simpler \che procedure compared to \advancedcheck, and, although it has a worse asymptotic message cost, it performs well in practice.

In \simplecheck, each subquorum $S_i$ has $\log \log n$ nodes that are chosen uniformly at random with replacement from the nodes in $U_i$ in the quorum path, for $1\leq i \leq \ell$.

First, node $\sender$ broadcasts the message to all nodes in $Q_1$, which send the message to the nodes of $S_1$. 
Then the nodes of $S_1$ forward this message to  the nodes of $S_\ell$ through a path of subquorums in the quorum path via all-to-all communication. 
Once the nodes in $S_\ell$ receive the message, they broadcast this message to all nodes in $Q_\ell$.
Further, if any node receives inconsistent messages or fails to receive an expected message, it initiates a call to \update.

Now we show that if the message was corrupted during the last call to \sendpath, the probability that \simplecheck fails to detect a corruption is $o(1)$ when $\ell = O(\log{n})$.

\begin{lemma}
If $\ell = O(\log{n})$, then \simplecheck fails to detect any message corruption with probability $o(1)$.
\end{lemma}

\begin{proof}
\simplecheck succeeds in detecting the message corruption if every subquorum has at least one good node.
We know that at least $1/2$-fraction of the nodes in any quorum are unmarked, then the probability that an unmarked bad node is selected uniformly at random is at most $1/4$.
Thus the probability of any subquorum having only bad nodes is at most
$(1/4)^{\log\log{n}} = 1/\log^2{n}$.  
Union-bounding over all $\ell$ subquorums, the probability of \simplecheck failing is at most $\ell/\log^2{n}$.  
For $\ell = O(\log{n})$, the probability that \simplecheck fails is $o(1)$. \qed
\end{proof}

\subsection{\advancedcheck}\label{app:check}

\begin{algorithm}[h]
\caption{$\textsc{CHECK2} (m, \receiver)$}
\label{a:check}


\textbf{Initializations:}  
node $\sender$ generates public/private key pair $k_{p},k_{s}$ to be used in this procedure; also let the subquorums $S_1, S_2,...,S_\ell$ be initially empty.

\begin{algorithmic}
\For{$i \gets 1, \ldots, \chkrows$}
\begin{enumerate}
\item Node $\sender$ constructs $R$ to be an $\ell$ by $|Q_{max}|$ array of random numbers, where $|Q_{max}|$ is the maximum size of any quorum.  
Note that $R[j,k]$ is a uniformly random number between $1$ and $k$.

\item Node $\sender$ sets $m'$ to be $m$, $k_p$, $\receiver$, and $R$ signed by $k_{s}$.

\item Node $\sender$ broadcasts $m'$ to all nodes in $Q_1$.

\item The nodes in $Q_1$ calculate the node, $x_{1} \in U_1$, to be added to $S_1$ using the number $R[1,|U_1|]$ to index $U_1$'s nodes sorted by their IDs.

\item The nodes in $Q_1$ send $m'$ to the nodes of $S_1$.

\item For $j \gets 1, \ldots, \ell-1$ do
\begin{enumerate}
\item All $i$ nodes in $S_{j}$ calculate the node, $x_{j+1} \in U_{j+1}$, to be added to $S_{j+1}$ using the number $R[j+1,|U_{j+1}|]$ to index $U_{j+1}$'s nodes sorted by their IDs.

\item The nodes in $S_j$ send $m'$ to node $x_{j+1}$.

\item Node $x_{j+1}$ sends $m'$ to all the nodes in $S_{j+1}$. 
\end{enumerate}

\item The nodes in $S_{\ell}$ broadcast $m'$ to all nodes in $Q_\ell$.

\end{enumerate}
\EndFor
\end{algorithmic}
\textbf{Note that:}
throughout this procedure, if a node has previously received $k_{p}$, then it verifies each subsequent message with it; also if a node receives inconsistent messages or fails to receive and verify an expected message, then it initiates a call to \update. 
\end{algorithm}

In this section, we describe \advancedcheck algorithm, which is stated formally as Algorithm~\ref{a:check}.  
Firstly, node $\sender$ generates a public/private key pair $k_{p},k_{s}$ to let the nodes verify any message received.
Then \advancedcheck runs for $4\log^{*}n$ rounds, and has a subset $S_{j} \subset U_j$ for each quorum $Q_{j}$ in the quorum path, for $1\leq j \leq \ell$.
Every $S_{j}$ is an incremental subquorum, where each $S_{j}$ is initially empty; and in each round, an unmarked node, $x_j$, is selected uniformly at random from all nodes in $U_{j}$ and is added to $S_{j}$, for all $1 \leq j \leq \ell$.

In each round, node $\sender$ broadcasts the message to all nodes in $Q_1$, which forward such message to the nodes of $S_1$.
Then for all $1 \leq j < \ell$, all nodes in $S_j$ send the message to node $x_{j+1}$, which forwards the message to all nodes in $S_{j+1}$.
Finally, all nodes in $S_{\ell}$ broadcast the message to all nodes in $Q_\ell$.

Note that if any node receives inconsistent messages or fails to receive and verify any expected message in any round, it initiates a call to \update.

\begin{figure}
\centerline{\includegraphics[scale=0.2]{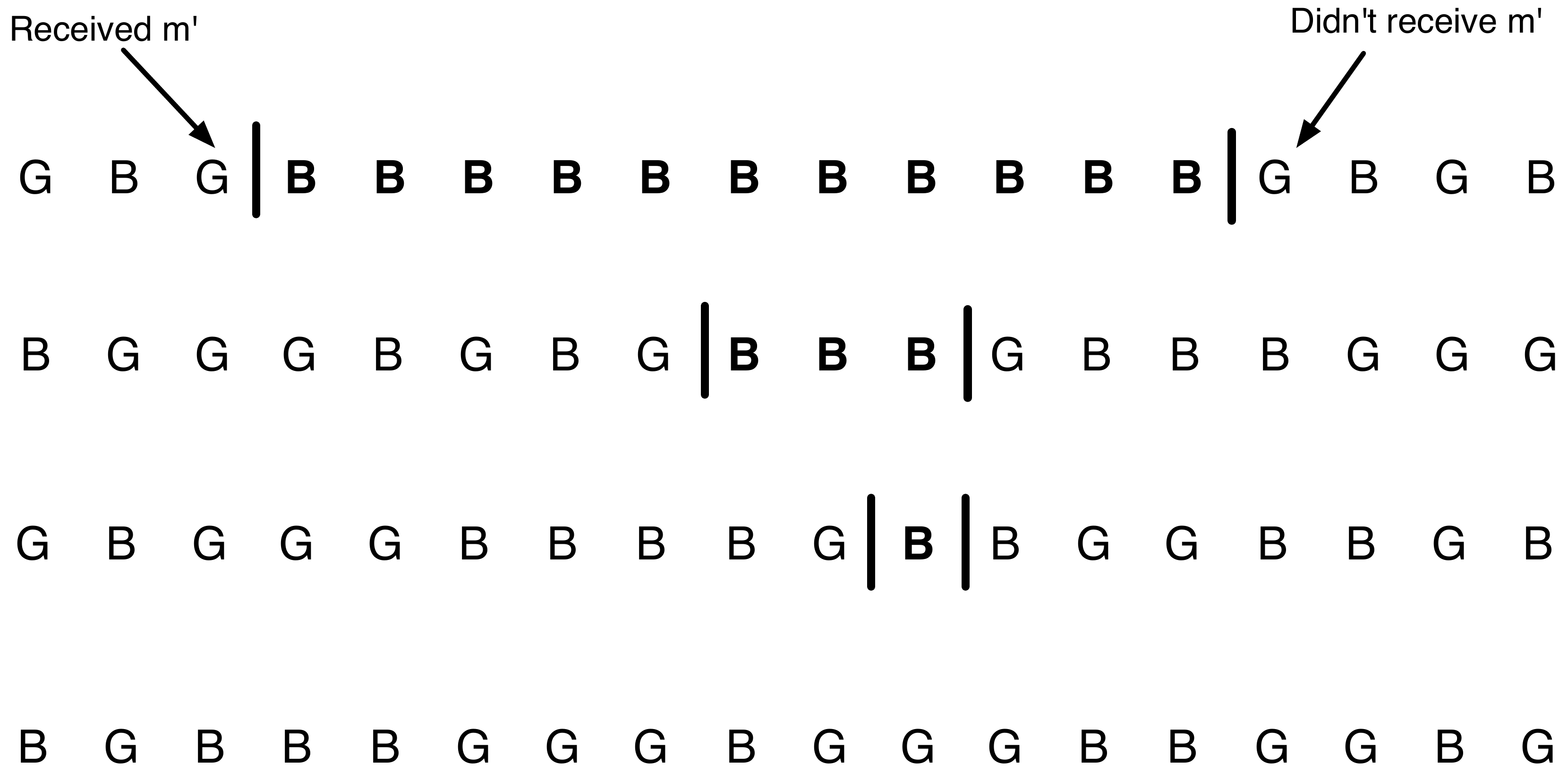}}
\caption{Example run of \advancedcheck}
\label{f:check}
\end{figure}

An example run of \advancedcheck is illustrated in Figure~\ref{f:check}.  
In this figure, there is a column for each quorum in the quorum path and a row for each round of \advancedcheck.  
For a given row and column, there is a G or B in that position depending on whether the node selected in that particular round and that particular quorum is good (G) or bad (B).

Recall that the adversary knows \advancedcheck algorithm, but in any round, the adversary does not know which nodes will be selected to participate for any subsequent round. 
The adversary's strategy is to maintain an interval of bad nodes in each row to corrupt (or drop) the message  so that \advancedcheck will not be able to detect the corruption.

In Figure \ref{f:check}, the intervals maintained by the adversary are outlined by a left and a right bar in each row; where the best interest of the adversary is to maintain the longest interval of bad nodes in the first row.
The left bar in each row specifies the rightmost subquorum in which there is some good node that receives $m'$.  
The right bar in each row specifies the leftmost subquorum in which there is some good node that does not receive $m'$.

Note that, as rounds progress, the left bar can never move leftwards, because a node that has already received $k_{p}$ will call \update unless it receives messages signed with $k_{p}$ for all subsequent rounds.
Note further that the right bar can never move rightwards, since the nodes of each subset $S_j$ send $m'$ to the node $x_{j+1}$, which forwards such message to all nodes that are currently in $S_{j+1}$, for $1\leq j < \ell$.
Finally, when these two bars meet, a corruption is detected.

Intuitively, the reason \advancedcheck requires only $4\log^{*}n$ rounds is because of a probabilistic result on the maximum length run in a sequence of coin tosses.  
In particular, if we have a coin that takes on value ``B'' with probability at most $1/4$, and value ``G'' with probability at least $3/4$, and we toss it $x$ times, then the expected length of the longest run of B's is $\log x$.  
Thus, if in some round, the distance between the left bar and the right bar is $x$, we expect in the next round this distance will shrink to $\log x$.  
Intuitively, we might expect that, if the quorum path is of length $\ell$, then $O(\log^{*} \ell)$ rounds will suffice before the distance shrinks to $0$.
This intuition is formalized in Lemmas \ref{l:intervalbad} and \ref{l:check} of Section~\ref{s:analysis}.

In contrast to \simplecheck, even though all nodes in column $9$ are bad in Figure \ref{f:check}, \advancedcheck algorithm handles this case since the good node in row $3$ and column $10$ receives the message $m'$.

\begin{algorithm}[h]
\caption{$\update$}
\label{a:update}

\begin{enumerate}
\item The node, $x$, making the call to \update broadcasts this fact to all nodes in its quorum, $Q'$, along with all the messages that $x$ has received during this call to \send.  The nodes in $Q'$ verify that $x$ received inconsistent messages before proceeding.
\item The quorum $Q'$ propagates the fact that a call to \update is occurring, via all-to-all communication, to all quorums $Q_{1}, Q_{2}, \ldots, Q_{\ell}$.
\item Each node involved in the last call to \sendpath or \che, except node $\sender$ and node $\receiver$, compiles all messages they have received (and from whom) in that last call, and broadcasts all these messages to all nodes in its quorum and the neighboring quorums.
\item Each node involved in the last call to \sendpath, except node $\sender$ and node $\receiver$, broadcasts all messages they have sent (and to whom) in that last call, to all nodes in its quorum and the neighboring quorums.
\item A node $v$ is \emph{in conflict} with a node $u$ if:
\begin{enumerate}
\item node $u$ was scheduled to send a message to node $v$ at some point in the last call to \sendpath or \che; and
\item node $v$ does not receive an expected message from node $u$ in step 3 or step 4, or node $v$ receives a message in step 3 or step 4 that is different than the message that it has received from node $u$ in the last call to \sendpath or \che.
\end{enumerate}
\item For each pair of nodes, $(u,v)$, that are in conflict in ($Q_k$, $Q_{k+1}$), for $1\leq k <\ell$:
\begin{enumerate}
\item node $v$ sends a \emph{conflict} message ``(u,v)" to all nodes in $Q_v$, 
\item each node in $Q_v$ forwards this conflict message to all nodes in $Q_v$ and all nodes in $Q_u$, 
\item quorum $Q_u$ (or $Q_v$) sends the conflict message to all other quorums that node $u$ (or $v$) is in, and 
\item each quorum that node $u$ or node $v$ is in sends such conflict message to its neighboring quorums.
\end{enumerate}
\item Each node that receives this conflict message mark the nodes $u$ and $v$.
\item If the half of nodes in any quorum have been marked, they are set unmarked in all quorums these nodes are in, and their neighboring quorums are notified.
\end{enumerate}
\end{algorithm}

\subsection{\update}

When a message corruption occurs and \che detects this corruption, \update is called. 
Now the task of \update is to 
1) determine the location in which the corruption occurred; 
2) mark the nodes that are in conflict.
The \update algorithm is described formally as Algorithm~\ref{a:update}. When \update starts, all nodes in each quorum in the quorum path are notified.

To determine the location in which the corruption occurred, each node previously involved in \sendpath or \che broadcasts to all nodes in its quorum and the neighboring quorums all messages they have received in the previous call to \sendpath or \che. 
Moreover, to announce the unmarked nodes that are selected uniformly at random in the last call to \sendpath, every node involved in the last call to \sendpath, broadcasts all messages they have sent (and to whom) in such call, to all nodes in its quorum and the neighboring quorums.

We say that a node $v$ is \emph{in conflict} with a node $u$ if 
1) node $u$ was scheduled to send a message to node $v$ at some point in the previous call to \sendpath or \che; and 2) node $v$ does not receive an expected message from node $u$ in this call to \update, or node $v$ receives a message in this \update that is different than the message that it has received from node $u$ in the previous call to \sendpath or \che.

For each pair of nodes, $(u,v)$, that are in conflict in quorums ($Q_k, Q_{k+1}$) for $1\leq k < \ell$, these two quorums send a \emph{conflict} message ``(u,v)" to all quorums in which node $u$ or node $v$ is and to all neighboring quorums to mark these nodes.
Recall that if the half of nodes in any quorum are marked, this quorum notifies all other quorums in which these nodes are via all-to-all communication to unmark such nodes, and their neighboring quorums are notified as well.

\section{Analysis} \label{s:analysis}
In this section, we prove the lemmas required for Theorem~\ref{thm:corruptions}. Throughout this section, \sendpath calls \advancedcheck.
Also we let all logarithms be base 2.

\begin{lemma}\label{l:intervalbad}
When a coin is tossed $x$ times independently given that each coin has a tail with probability at most $1/4$, then the probability of having any substring of length $\max(1, \log{x})$ being all tails is at most $1/2$.
\end{lemma}

\begin{proof}
The probability of a specific substring of length $\log{x}$ being all tails is 
$$\left(\frac{1}{4}\right)^{\log{x}} = \frac{1}{x^2}.$$
Union bounding over all possible substrings of length $\log{x}$, then the probability of any all-tailed substring existing is at most $x\frac{1}{x^2}$; or equivalently, for $x \geq 2$, $\frac{1}{x} \le \frac{1}{2}$; and for $x=1$, the probability of having a substring of $\max(1, \log{x})$ tail is trivially $1/4$.  \qed 
\end{proof}

The next lemma shows that the algorithm \advancedcheck catches corruptions with probability at least $1/2$.

\begin{lemma}\label{l:check}
Assume some node selected uniformly at random in the last call to \sendpath has corrupted a message.  Then when the algorithm \advancedcheck is called, with probability at least $1/2$, some node will call \update.
\end{lemma}


\begin{lemma}\label{l:update}
If some node selected uniformly at random in the last call to \sendpath has corrupted a message, then the algorithm \update will identify a pair of neighboring quorums $Q_{j}$ and $Q_{j+1}$, for some $1 \leq j < \ell$, such that at least one pair of nodes in these quorums is in conflict and at least one node in such pair is bad.
\end{lemma}

\begin{proof}

First we show that if a pair of nodes $x$ and $y$ is in conflict, then at least one of them is bad.  Assume not. Then both $x$ and $y$ are good.  Then node $x$ would have truthfully reported what it received; any message that $x$ received would have been sent directly to $y$; and $y$ would have truthfully reported what it received from $x$.  But this is a contradiction, since for $x$ and $y$ to be in conflict, $y$ must have reported that it received from $x$ something different than what $x$ reported receiving. 

Now consider the case where a selected unmarked bad node corrupted a message in the last call to \sendpath. 
By the definition of corruption, there must be two good nodes $q_j$ and $q_k$ such that $j < k$ and $q_j$ received the message $m'$ sent by node \sender, and $q_k$ did not. We now show that some pair of nodes between $q_j$ and $q_k$ will be in conflict.  
Assume this is not the case.  Then for all $x$, where $j \leq x < k$, nodes $q_x$ and $q_{x+1}$ are not in conflict.  But then, since node $q_j$ received the message $m'$, and there are no pairs of nodes in conflict, it must be the case that the node $q_k$ received the message $m'$. This is a contradiction.  
Thus, \update will find two nodes that are in conflict, and at least one of them will be bad.

Now we prove that at least one pair of nodes is found to be in conflict as a result of calling \update.  
Assume that no pair of nodes is in conflict.  Then for every pair of nodes $x$ and $y$, such that $x$ was scheduled to send a message to $y$ during any round $i$ of \advancedcheck, $x$ and $y$ must have reported that they received the same message in round $i$.  In particular, this implies via induction, that for every round $i$, for all $j$, where $1 \leq j \leq \ell$, all nodes in the sets $S_{j}$ must have broadcasted that they received the message $m'$ that was initially sent by node $\sender$ in round $i$.  But if this is the case, the node $x$ that initially called \update would have received no inconsistent messages.  This is a contradiction since in such a case, node $x$ would have been unsuccessful in trying to initiate a call to \update.  Thus, some pair of nodes must be found to be in conflict, and at least one of them is bad.
\qed


\end{proof}

The next lemma bounds the number of times that \update must be called before all bad nodes are marked.

\begin{lemma}\label{l:numc}
\update is called at most $3t/2$ times before all bad nodes are marked.
\end{lemma}

\begin{proof}
By Lemma~\ref{l:update}, if a message corruption occurred in the last call to \sendpath, and is caught by \advancedcheck, then \update is called. \update identifies at least one pair of nodes that are in conflict. 

Now let $g$ be the number of good nodes that are marked, and let $b$ be the number of bad nodes that are marked.
Also let $f(g,b) = b - g/3$.

For each corruption caught,
at least one bad node is marked, and so $f(g,b)$ increases by at least $2/3$ since $b$ increases by at least $1$ and $g$ increases by at most $1$.
When a $1/2$-fraction of nodes in any quorum $Q$ of size $|Q|$ are unmarked, 
$f(g,b)$ further increases by at least $0$ since $g$ decreases by at least $\frac{3|Q|}{8}$ and $b$ decreases by at most $\frac{|Q|}{8}$.

Hence, $f(g,b)$ is monotonically increasing by at least $2/3$ for each corruption caught.
When all bad nodes are marked, $f(g, b) \leq t$.
Therefore, after at most $3t/2$ calls to \update, all bad nodes are marked. 
\qed
\end{proof}

\section{Empirical Results}
\label{s:empirical}

\subsection{Setup}

In this section, we empirically compare the message cost and the fraction of messages corrupted of two algorithms via simulation.  
The first algorithm we simulate is \bfly algorithm from \cite{FS}.  This algorithm has no self-healing properties, and simply uses all-to-all communication between quorums that are connected in a butterfly network.  
The second algorithm is \llog, wherein we apply our self-healing algorithm in the butterfly networks using \simplecheck.

In our experiments, we consider two butterfly network sizes: $n = 14{,}116$ and $n = 30{,}509$, where $\ell = \lfloor\log n\rfloor - 2$, the quorum size is $\lfloor 4\log n\rfloor$ and the subquorum size is $\lfloor\log \log n\rfloor$.
Moreover, we do our experiments for several fractions of bad nodes such as $f$ equal to $1/8, 1/16, 1/32$ and $1/64$, where $f = t/n$.

Our simulations consist of a sequence of calls to \emph{SEND} over the network, given a pair of nodes $\sender, \receiver$, chosen uniformly at random, such that node $\sender$ sends a message to node $\receiver$. 
We simulate an adversary who chooses at the beginning of each simulation a fixed number of nodes to control uniformly at random without replacement.  
Our adversary attempts to corrupt messages between nodes whenever possible. 
Aside from attempting to corrupt messages, the adversary performs no other attacks.

\subsection{Results}
The results of our experiments are shown in Figures \ref{fig:msgs} and \ref{fig:corr}.
Our results highlight two strengths of our self-healing algorithms (\llog) when compared to algorithms without self-healing (\bfly).
First, the message cost per \send decreases as the total number of calls to \send increases, as illustrated in Figure \ref{fig:msgs}. Second, for a fixed number of calls to \send, the message cost per \send decreases as the total number of bad nodes decreases, as shown in Figure \ref{fig:msgs} as well. In particular, when there are no bad nodes, \llog has dramatically less message cost than \bfly.

Figure \ref{fig:msgs} shows that for $n = 14{,}116$, the number of messages per SEND for \bfly is $30{,}516$; and for \llog, it is $525$. 
Hence, the message cost is reduced by a factor of $58$.
Also for $n = 30{,}509$, the number of messages per SEND for \bfly is $39{,}170$; and for \llog, it is $562$; which implies that the message cost is reduced by a factor of $70$.

In Figure \ref{fig:corr}, \bfly has $0$ corruptions; however, for \llog, the fraction of messages corrupted per \send decreases as the total number of calls to \send increases. Also for a fixed number of calls to \send, the fraction of messages corrupted per \send decreases as the total number of bad nodes decreases. 

Furthermore, in Figure \ref{fig:corr}, for each network, given the size and the fraction of bad nodes, if we integrate the corresponding curve, we get the total number of times that a message can be corrupted in calls to \send in this network. 
These experiments show that the total number of message corruptions is at most $3t(\log \log n)^2$.

\begin{figure*}
\centerline{
\includegraphics[scale=0.35]{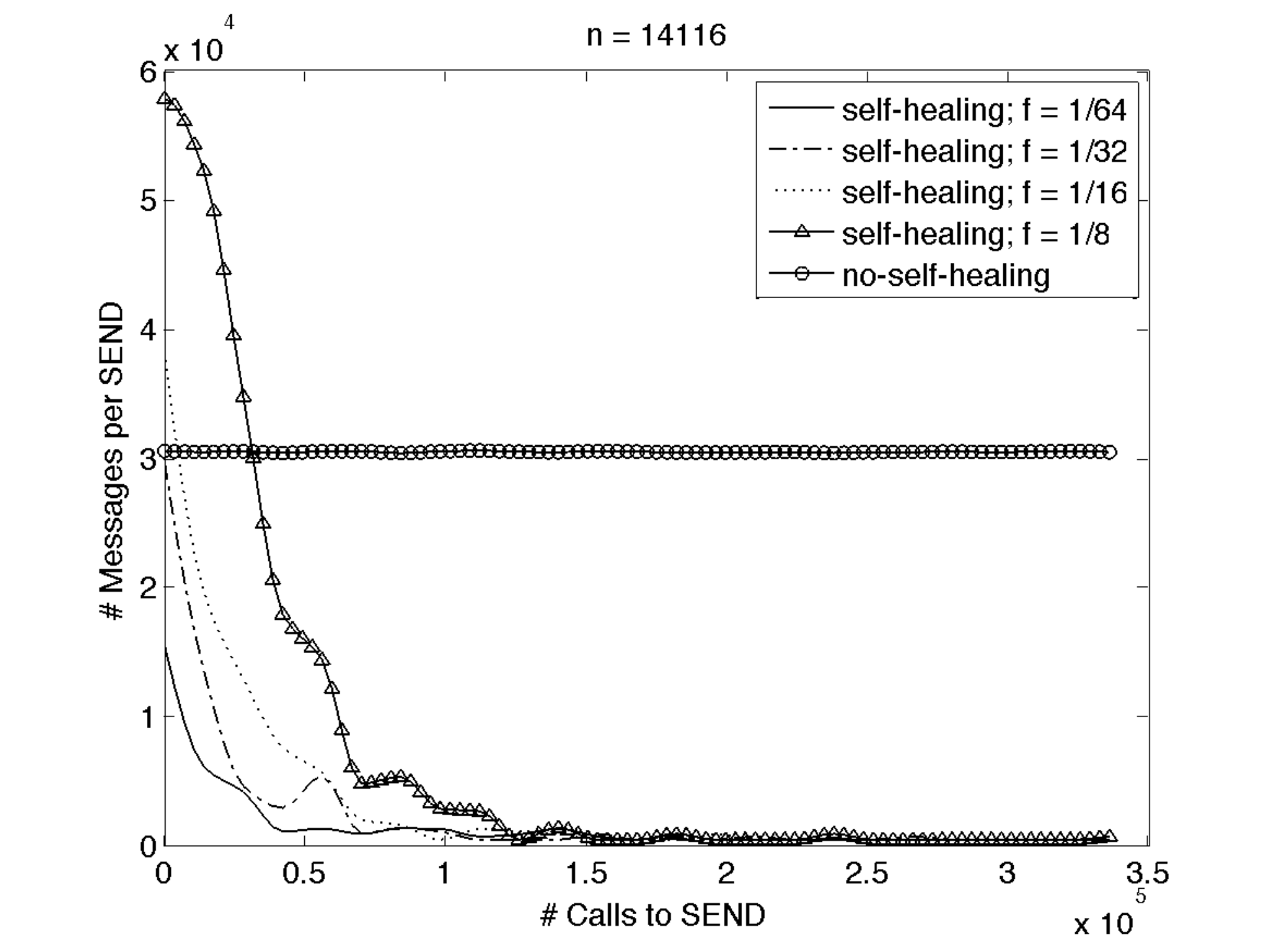}
\includegraphics[scale=0.35]{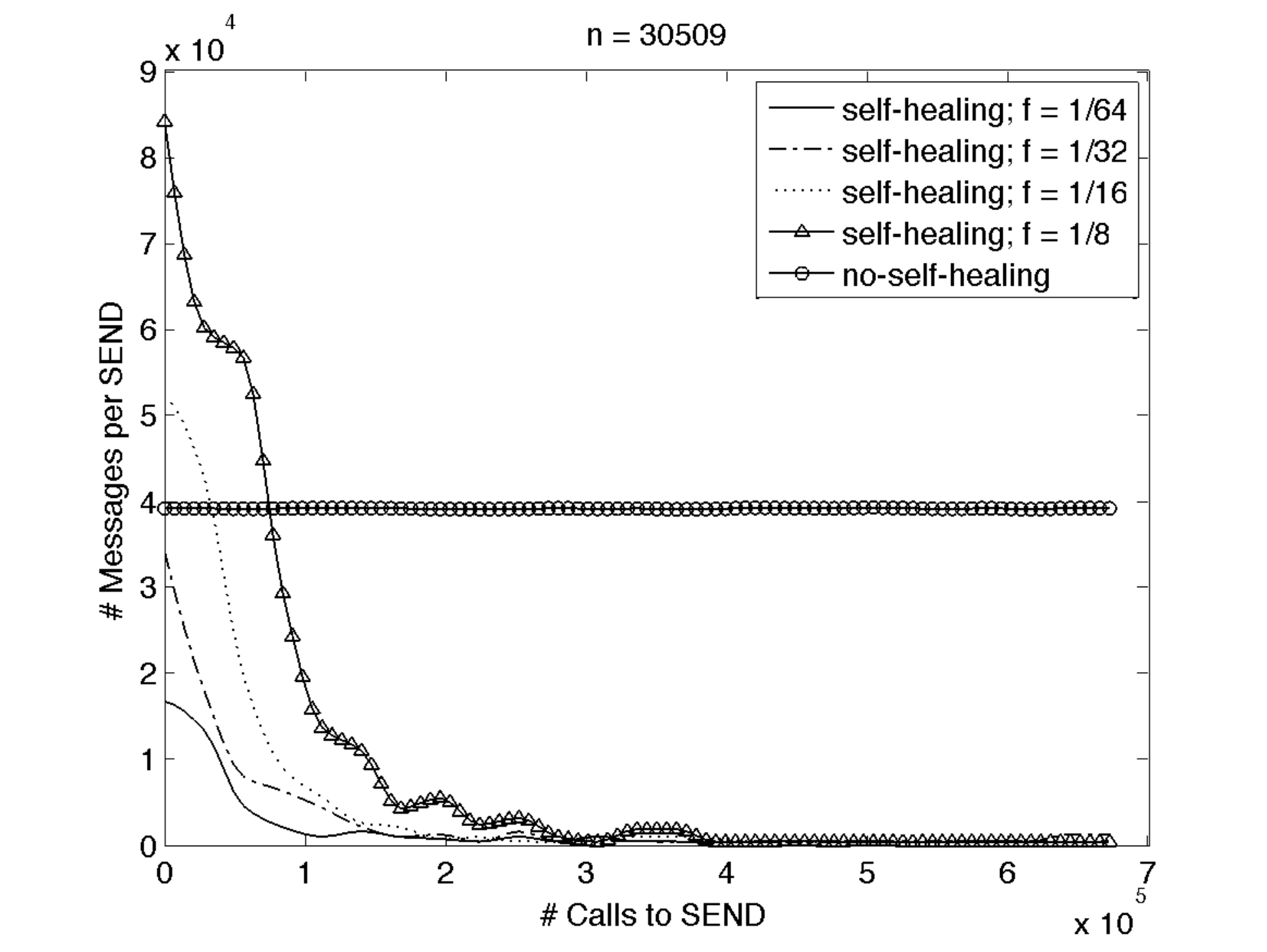}
}
\caption{\# Messages per \send versus \# calls to \send, for $n = 14{,}116$ and $n = 30{,}509$.}
\label{fig:msgs}
\end{figure*}

\begin{figure*}
\centerline{
\includegraphics[scale=0.35]{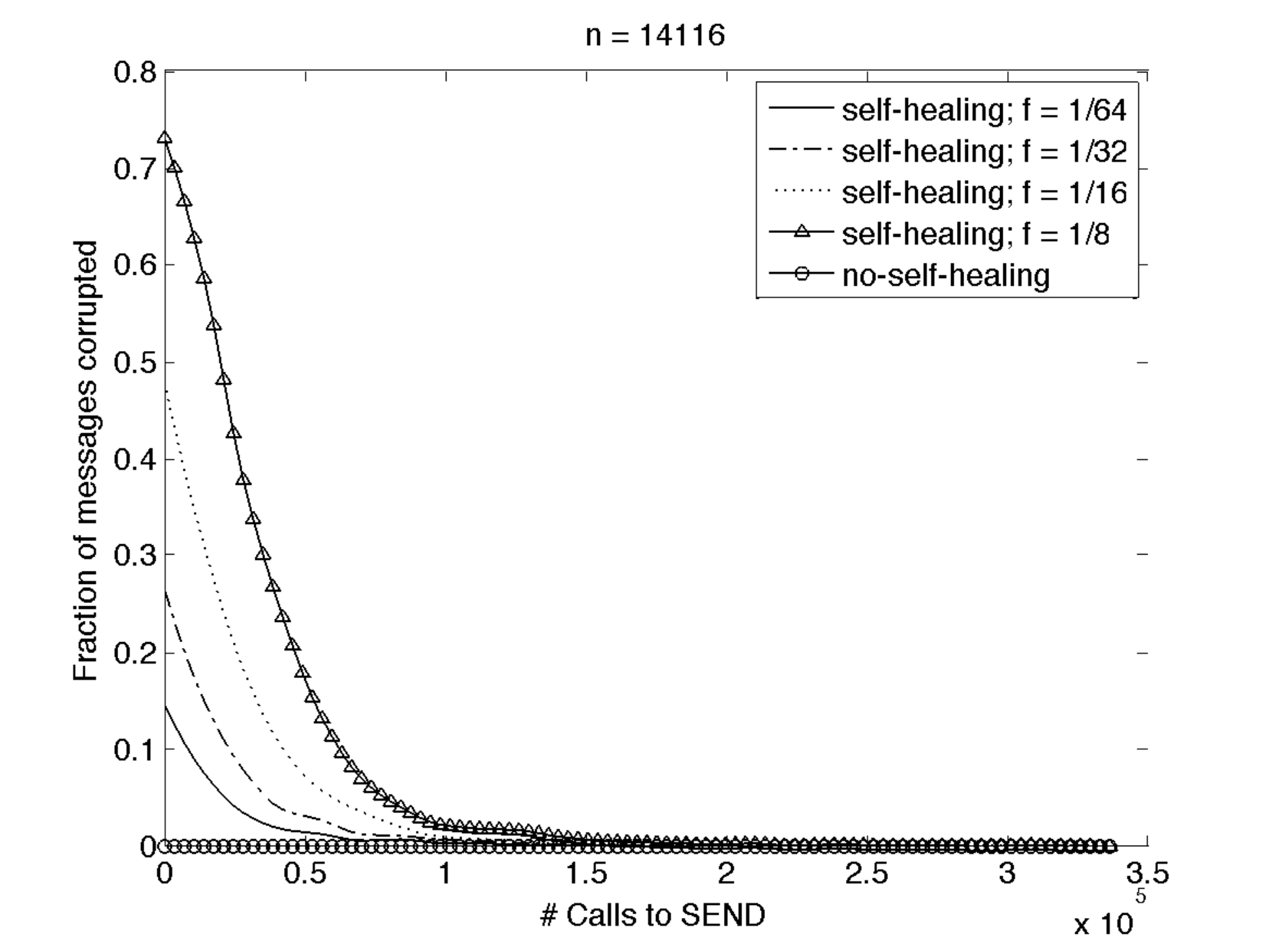}
\includegraphics[scale=0.35]{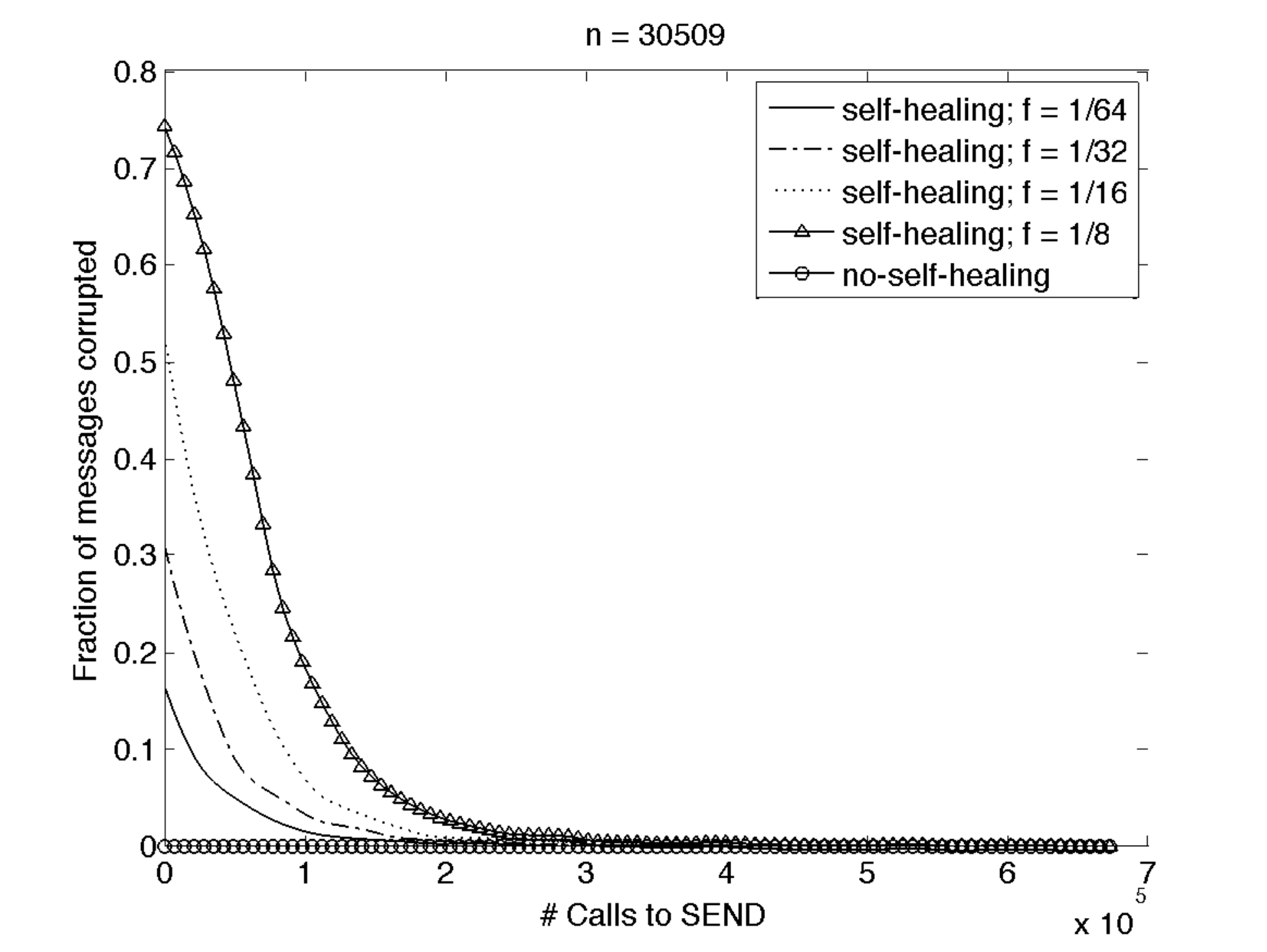}
}
\caption{Fraction of messages corrupted versus \# calls to \send, for $n = 14{,}116$ and $n = 30{,}509$.}
\label{fig:corr}
\end{figure*}


\section{Conclusion and Future Work}\label{s:conc}

We have presented algorithms that can significantly reduce communication cost in attack-resistant peer-to-peer networks.
The price we pay for this improvement is the possibility of message corruption.  In particular, if there are $t \leq \numbadnodes$ bad nodes in the network, our algorithm allows $O(t \chkprobinv)$ message transmissions to be corrupted in expectation.  

Many problems remain.  
First, it seems unlikely that the smallest number of corruptions allowable by an attack-resistant algorithm with optimal message complexity is $O(t \chkprobinv)$.  Can we improve this to $O(t)$ or else prove a non-trivial lower bound?  Second, can we apply techniques in this paper to problems more general that enabling secure communication?  For example, can we create self-healing algorithms for distributed \emph{computation} with Byzantine faults?  Finally, can we optimize constants and make use of heuristic techniques in order to significantly improve our algorithms' empirical performance?

\bibliographystyle{plain}
\bibliography{markalgorithm}

\clearpage

\appendix
\section{Appendix - Deferred Proofs}\label{app:proofoftheorem}

\noindent
{\bf Lemma \ref{l:check}.} 
Assume some node selected uniformly at random in the last call to \sendpath has corrupted a message.  Then when the algorithm \advancedcheck is called, with probability at least $1/2$, some node will call \update.

\begin{proof}
This proof makes use of the following two propositions, but first we define the critical intervals in \advancedcheck.

\begin{definition}
We say that there is a $(j, k)$ critical interval in round $i$ if
a reliable transmittal in round $i$ from a good node in $Q_j$ to a good node in $Q_k$ will result in a call to \update.
\end{definition}

\begin{fact}\label{fact:transmits_reliably}
If $(j, k)$ is a critical interval in round $i$ then $\exists j', k', j\leq j'\leq k' \leq k$ such that $(j', k')$ is a critical interval in round $i+1$.

%
\end{fact}

To prove this proposition, note that the good nodes in $Q_{j}$ that are chosen by \advancedcheck in rounds $i$ or less know $\sender$'s public key, $k_{p}$.  Thus they must receive uncorrupted messages signed by $\sender$'s private key, $k_{s}$, in all rounds subsequent to $i$; otherwise, \update will be called.

\begin{fact}\label{fact:allmustbebad}
Fix a round $i+1$ in \advancedcheck, and let $(j', k')$ be a critical interval in round $i+1$ that minimizes the value $k'-j'$.  Then in round $i+1$, \update is called unless all nodes that are chosen between quorums $Q_{j'}$ and $Q_{k'}$ are bad.
\end{fact}


To show Proposition \ref{fact:allmustbebad}, assume by way of contradiction that it is false. Then there exists $x'$ such that $j' < x' < k'$, where the node $p'$ chosen in round $i+1$ at $Q_{x'}$ is good. There are two cases for what happens in round $i+1$:
\begin{itemize}
\item Case 1: The node $p'$ receives the message $m'$ sent by node \sender. But then $(j', x')$ is a critical interval in round $i+1$ and $x' - j' < k' - j'$. 
This contradicts the assumption that the indices $j'$ and $k'$ had the minimal distance among all critical intervals in round $i+1$.
\item Case 2: The node $p'$ does not receive the message $m'$ sent by node \sender.  But then $(x', k')$ is a critical interval in round $i+1$ and $k' - x' < k' - j'$.  This again contradicts the assumption that the indices $j'$ and $k'$ had the minimal distance among all critical intervals in round $i+1$.
\end{itemize}

Now we can use these two propositions to prove the lemma. Let $X_{i}$ be an indicator random variable that it is equal to $1$ if $(k'-j') \leq \log(k-j)$ and $0$ otherwise, where $(j,k)$ is a critical interval in round $i$ and $(j',k')$ is a critical interval in round $i+1$.
Recall that at least $1/2$-fraction of the nodes in any quorum are unmarked, then the probability that an unmarked bad node is selected uniformly at random is at most $1/4$.
By Lemma~\ref{l:intervalbad}, the probability of having any substring of length $\max(1, \log{x})$ bad nodes in a sequence of $x$ nodes selected independently is at most $1/2$.
Thus, each $X_{i}$ is $1$ with probability at least $1/2$. 
We require at least $\log^* n$ of the $X_{i}$ random variables to be $1$ in order for some node to call \update.\footnote{Here we assume $\ell \leq n$.  
However, we can achieve the same asymptotic results assuming that $\ell$ is bounded by a polynomial in $n$.} 
Let $X = \sum_{i=1}^{\chkrows} X_i$.  Then $\E{(X)} = 2 \log^* n$, and since the $X_{i}$'s are independent, by Chernoff bounds,
$$
\Pr\left(X < (1-\delta)2\log^*n\right) \leq \left(\frac{e^\delta}{(1+\delta)^{1+\delta}}\right)^{2\log^*n}.
$$

For $\delta = \frac{1}{2}$ and $n > 16$, 
$
\Pr\left(X < \log^*n\right) \leq \left(\frac{e^\frac{1}{2}}{(\frac{3}{2})^{\frac{3}{2}}}\right)^{2\log^*n} < \frac{1}{2}.
$

Thus the probability that \advancedcheck succeeds in finding a corruption and calling \update is at least $1/2$. \qed

\end{proof}

\noindent
{\bf Proof of Theorem \ref{thm:corruptions}.}
We first show the message complexity and the latency of our algorithms.
By Lemma \ref{l:numc}, the number of calling \update is at most $3t/2$. Thus the resource cost of all calls to \update are bounded as the number of calls to \send grows large.
Therefore, for the amortized cost, we consider only the cost of the calls to \sendpath and \advancedcheck.

When sending a message through $\ell$ quorums, \sendpath has message cost $O(\ell + \log n)$ and latency $O(\ell)$.
\advancedcheck has a message cost of $O((\ell + \log n) (\log^* n)^2)$ and a latency of $O(\ell \log^* n)$, but \advancedcheck is called only with probability $1/\chkprobinv$.  
Hence, the call to \send has amortized expected message cost $O(\ell + \log n)$ and amortized expected latency $O(\ell)$.

More specifically, if we perform any number of message sends through quorum paths, where $\ell_{M}$ is the longest such path, and $\mathcal{L}$ is the sum of the quorums traversed in all such paths, then the expected total number of messages sent will be $O(\mathcal{L} + t \cdot ( \ell_{M} \log^{2} n + \log^5 n))$, and the latency is $O(t \cdot \ell_M)$. 
This is true since each call to \update has message cost $O(\ell_M \log^2 n + \log^5 n)$ and latency $O(\ell_M)$, where 
1) the node, $x$, making the call to \update broadcasts its reason of calling \update to all nodes in its quorum, this has message cost $O(\log n)$ and latency $O(1)$;
2) all nodes in every quorum in the quorum path are notified via all-to-all communication when \update is called, these notifications have a message cost of $O(\ell_M \log^2 n)$ and a latency of $O(\ell_M)$; 
3) \update has $O(\log^{*}n)$ broadcasts over at most $\ell_M$ quorums, that has message cost $O(\ell_M \log^{*}n \cdot \log n)$ and latency $O(1)$; 
4) the message cost when all nodes in $Q_1$ and $Q_\ell$ broadcast is $O(\log^2 n)$ with latency $O(1)$;
5) for marking a pair of nodes that are in conflict, the message cost is $O(\log^3 n)$ and the latency is $O(1)$; and 
6) we know that marking a pair of nodes that are in conflict could cause $O(\log n)$ quorums to be unmarked; and unmarking half of nodes in a quorum has a message cost of $O(\log^4 n)$. 
Then unmarking $O(\log n)$ quorums has a message cost of $O(\log^5 n)$ and a latency of $O(1)$. 

Recall that by Lemma~\ref{l:numc}, the number of times \advancedcheck must catch corruptions before all bad nodes are marked is at most $3t/2$.  
In addition, if a bad node caused a corruption during a call to \sendpath, then by Lemmas~\ref{l:check} and~\ref{l:update}, with probability at least $1/2$, \advancedcheck will catch it.  
As a consequence, it will call \update, which marks the nodes that are in conflict.  $\update$ is thus called with probability $1/\chkprobinv$, so the expected total number of corruptions is at most $3t \chkprobinv$.

\end{document}